\documentclass{article}
\usepackage{epsfig}
\usepackage{graphicx}
\usepackage{mathrsfs}
\usepackage{graphicx}
\usepackage{amsmath}
\usepackage{amsfonts}
\usepackage{amssymb}
\usepackage{amsthm}

\newtheorem{theorem}{Theorem}[section]
\newtheorem{lemma}{Lemma}[section]

\begin{document}
\title{Hohenberg-Kohn Theorem for Coulomb Type Systems  \thanks{{This work was partially
supported by the National Science Foundation of China under grants
10871198 and 10971059, the Funds for Creative Research Groups of
China under grant 11021101, and the National Basic Research Program
of China under grant 2011CB309703.}}}

\author{ Aihui
Zhou\thanks{LSEC, Institute of Computational Mathematics and
Scientific/Engineering Computing, Academy of Mathematics and Systems
Science, Chinese Academy of Sciences, Beijing 100190, China
({azhou@lsec.cc.ac.cn}).}}

\date{}
\maketitle

\begin{abstract}
Density functional theory (DFT) has become a basic tool for the
study of electronic structure of matter, in which the Hohenberg-Kohn
theorem plays a fundamental role in the development of DFT.
Unfortunately,  the existing proofs are incomplete even
incorrect; besides, the statement of the Hohenberg-Kohn theorem for
many-electron Coulomb systems is  not perfect. In this paper, we
shall restate the Hohenberg-Kohn theorem for Coulomb type systems
and present a rigorous proof by using  the Fundamental Theorem of
Algebra.
\end{abstract}


\section{Introduction}\setcounter{equation}{0}
 The modern formulation of density functional theory (DFT)
originated in the work  of Hohenberg and Kohn
\cite{hohenberg-kohn64}, on which based the other classic work in
this field by Kohn and Sham \cite{kohn-sham65}, the Kohn-Sham
equation, has become a basic mathematical model of much of
present-day methods for treating electrons in atoms, molecules,
condensed matter, and man-made structures
\cite{eschrig96,gross-dreizler95,jensen99,martin04,parr-yang89}.
Although it is quite profound, DFT is  not entirely elaborated yet
(c.f., e.g.,
\cite{lieb83,pino-bokanowski-etal07,szczepanik-dulak-wesolowski07}
and references cited therein).
 Since the relevant assumptions are incompatible  with the
Kato cusp condition, Kryachko has pointed out that the usual {\it
reductio ad absurdum} proof of  the original Hohenberg-Kohn theorem
is unsatisfactory \cite{kryachko05,kryachko06}. Note that Kato
theorem \cite{kato57} tells the electron-nucleus cusp conditions at
nucleus positions only,  we are not able to uniquely determine the
electron density from the cusp conditions, though the electron density uniquely determines
the external Coulombic potential. Consequently, there is a
gap in the proof of the theorem by using the Kato theorem in
\cite{kryachko06}.
We note that Lieb has tried to examine the theorem rigorously
\cite{lieb83}. But Lieb's proof required that the $N$-particle
wavefunction does not vanish in a set of positive measure that is
unclear in a real system (c.f. \cite{pino-bokanowski-etal07}). We
refer to \cite{hadjisavvas-theophilou84,kryachko05,levy79,lieb83,
pino-bokanowski-etal07,szczepanik-dulak-wesolowski07} and references
cited therein for discussions on the Hohenberg-Kohn theorem. Indeed,
the statement of the original Hohenberg-Kohn theorem for
many-electron Coulomb systems is not  perfect (see Section
\ref{h-k-theorem}).

In this paper, we shall state the Hohenberg-Kohn theorem for Coulomb
type systems precisely (see Theorem \ref{hohenberg-kohn}) and
present a rigorous proof by using the Fundamental Theorem of Algebra
(see Section \ref{proof}).

\section{Hohenberg-Kohn theorem}\label{h-k-theorem}\setcounter{equation}{0}
We see that the approach of Hohenberg and Kohn is to formulate DFT
as an exact theory of many-body systems. The formulation applies to
any system of interacting particles in an external potential $v$,
including any problem of electrons and fixed nuclei, where the
Hamiltonian can be written as
\begin{eqnarray}\label{hamit-ext}
{\cal H}&=&-\sum_{i=1}^N\frac{\hbar^2}{2m_e}\nabla^2_{x_i}
+\sum^N_{i=1}v(x_i)+\frac{1}{2}\sum_{i,j=1, i\ne
j}^N\frac{e^2}{|x_i-x_j|},
\end{eqnarray}
where $\hbar$ is Planck's constant divided by $2\pi, m_e$ is the
mass of the electron, $\{x_i:i=1,\cdots,N\}$ are the variables that
describe the electron positions, and $e$ is the electronic charge.
For an electronic Coulomb system,
\begin{eqnarray}\label{coulomb-v}
v(x)\equiv
v_{\{_{Z_j}\},\{r_j\}}(x)=-\sum_{j=1}^M\frac{Z_je^2}{|x-r_j|}
\end{eqnarray} is determined by $\{Z_j:j=1,2,\cdots,M\}$, which are the valence
charges of the nuclei, and $\{r_j: j=1,2,\cdots,M\}$, which are the
positions of the nuclei. The energy of the system can be expressed
by
\begin{eqnarray}\label{energy}
E=(\Psi, {\cal H}\Psi)=(\Psi, (T+V_{ee})\Psi)+\int_{\mathbb{R}^3}
v(x)\rho(x),\end{eqnarray}
where$$T=-\sum_{i=1}^N\frac{\hbar^2}{2m_e}\nabla^2_{x_i}$$ is the
kinetic energy operator,$$ V_{ee}=\frac{1}{2}\sum_{i,j=1, i\ne
j}^N\frac{e^2}{|x_i-x_j|}$$ is the electron-electron repulsion
energy operator, and
\begin{eqnarray}\label{density-def}
\rho(x)\equiv
\rho^\Psi(x)=N\sum_{\sigma_1,\sigma_2,\cdots,\sigma_N}\int_{\mathbb{R}^{3(N-1)}}
|\Psi((x,\sigma_1),(x_2,\sigma_2),\cdots,(x_N,\sigma_N))|^2dx_2\cdots
dx_N \end{eqnarray} is the single-particle density.

Let ${\cal H}_0=T+V_{ee}$ and $v$ be a single-particle potential in
function space $$\mathbb{V}\equiv
L^{3/2}(\mathbb{R}^{3})+L^{\infty}(\mathbb{R}^{3}).$$ The total
Hamiltonian is ${\cal H}_v={\cal H}_0+{\cal V}$, where
$${\cal V}=\sum_{i=1}^Nv(x_i).$$
The associated ground state energy $E(v)$ is defined to be
\begin{eqnarray}\label{ground-state}
E(v)\equiv E(v,N)=\inf\{(\Psi,{\cal H}_v\Psi): \Psi\in
\mathscr{W}_N\},
\end{eqnarray} where
$$\mathscr{W}_N=\{\Psi\in H^1(\mathbb{R}^{3N}):
\sum_{\sigma_1,\sigma_2,\cdots,\sigma_N}\int_{\mathbb{R}^{3N}}|\Psi|^2=1\}.$$

Note that there may or may not be a minimizer $\psi$ in
$\mathscr{W}_N$, and if there is one it may not be unique
\cite{lieb83}. Thus, we should introduce a set  of
minimizers\begin{eqnarray*} \mathscr {G}_{v}\equiv \mathscr
{G}_{v,N}=\arg\inf\{(\Psi,{\cal H}_v\Psi): \Psi\in \mathscr{W}_N\}.
\end{eqnarray*}
Any $\Psi$ in $\mathscr {G}_{v}$ is called a ground state of
(\ref{ground-state}). Define\begin{eqnarray*}
 \mathscr {V}_N=\{v\in \mathbb{V}:  \mathscr {G}_{v,N}\not=\emptyset\}
\end{eqnarray*}
and \begin{eqnarray*}
 \mathscr {D}_N=\{\rho^\Psi:
 \Psi\in \mathscr {G}_{v} ~\mbox{for some}~ v\in \mathscr {V}_N\}.
\end{eqnarray*}
 If $\Psi\in \mathscr {G}_{v}$, then
\begin{eqnarray}\label{ground-state-weak}
\mathcal {H}_v\Psi=E(v)\Psi
\end{eqnarray}
in the distributional sense.

The original Hohnberg-Kohn theorem (see page B865 of
\cite{hohenberg-kohn64}) states that  the external potential $v$
``is a unique functional of" the electronic density in the ground
state,``apart from a trivial additive constant." In our notation,
Lieb's statement of this theorem may be written as the following
(see Theorem 3.2 of \cite{lieb83}):
\begin{theorem}\label{hohenberg-kohn0}
Suppose $\Psi_v\in \mathscr {G}_{v}$ and $\Psi_{v'}\in \mathscr
{G}_{v'}$. If $v\not=v'+constant$, then
$\rho^{\Psi_v}\not=\rho^{\Psi_{v'}}$.
\end{theorem}

Consider Coulomb potential set
\begin{eqnarray*}
\mathbb{V}_C=\left\{-\sum_{j=1}^M\frac{Z_je^2}{|x-r_j|}: Z_j\in
\mathbb{R}, r_j\in \mathbb{R}^3 (j=1,2,\cdots, M);
M=1,2,\cdots\right\}.\end{eqnarray*} Obviously,
$\mathbb{V}_C\subsetneq \mathbb{V}$.

It will be shown by the Fundamental Theorem of Algebra  that for any
$v,v'\in \mathbb{V}_C$, $v\not=v'+constant$ for any constant if
$v\not=v'$ (see Section \ref{proof} for details), which means that
the requirement $v\not=v'+constant$ in the Hohnberg-Kohn theorem or
Theorem \ref{hohenberg-kohn0} is superfluous for  electronic Coulomb
system. More precisely, $v\not=v'+constant$ in the Hohnberg-Kohn
theorem should be replaced by $v\not=v'$. Therefore the
Hohenberg-Kohn theorem or Theorem \ref{hohenberg-kohn0} for Coulomb
type systems should be restated as follows:
\begin{theorem}\label{hohenberg-kohn}
The map: $v\in \mathbb{V}_{C}\longrightarrow \rho_v\in \mathscr
{D}_N$ is one-to-one, where $\rho_v=\rho^{\Psi}$ with $\Psi\in
\mathscr {G}_{v}$.
\end{theorem}

\section{Lemmas}\label{lemmass}
\setcounter{equation}{0} To prove Theorem \ref{hohenberg-kohn}, we
need some lemmas.

\begin{lemma}\label{quadratic-lemma}
Let $n\ge 2$. If $t_j\in \mathbb{R}(j=1,2,\cdots,n)$ and
$$
\delta=\sum_{j=1}^nt_j,$$ then there exist non-zero polynomials
$\{H_{n,j}(s_1,s_2,\cdots,s_n): j=0,1,2,\cdots,2^{n-1}\}$ with real
coefficients satisfying
\begin{eqnarray}\label{quadratic-identity}
\sum^{2^{n-1}}_{j=0}H_{n,j}(t_1^2,t_2^2,\cdots,t_n^2)\delta^{2(2^{n-1}-j)}=0,
\end{eqnarray}
where $H_{n,j}(s_1,s_2,\cdots,s_n)(j=1,2,\cdots,2^{n-1})$ are
homogeneous:
\begin{eqnarray*}
H_{n,j}(\lambda s_1,\lambda s_2,\cdots,\lambda
s_n)&=&\lambda^jH_{n,j}(s_1,s_2,\cdots,s_n),~ \forall \lambda\in \mathbb{R}, j=1,2,\cdots,2^{n-1},\\
H_{n,0}(s_1, s_2,\cdots, s_n)&=& 1,
\end{eqnarray*}
and $H_{n,2^{n-1}}(s_1,s_2,\cdots,s_n)$ is a monic polynomial of
degree $2^{n-1}$.
\end{lemma}

\begin{proof}
 We prove the conclusion by induction on $n$. First, for
$n=2$, $t_1+t_2=\delta$, which implies
$$t_1^2+2t_1t_2+t_2^2=\delta^2$$
and hence
$$t_1^4+t_2^4-2t_1^2t_2^2-2\delta^2(t_1^2+t_2^2)+\delta^4=0.$$
Namely, (\ref{quadratic-identity}) is true for $n=2$.

For the induction step, suppose (\ref{quadratic-identity}) is true
for $n$. If $$\sum_{j=1}^nt_j+t_{n+1}=\delta,$$ then
$$\sum_{j=1}^{n}t_j=\delta-t_{n+1}.$$
By the induction hypothesis, we have that there exist non-zero
polynomials $\{H_{n,j}(s_1,s_2,\cdots,s_n):
j=0,1,2,\cdots,2^{n-1}\}$ with real coefficients satisfying
$H_{n,j}(s_1,s_2,\cdots,s_n)(j=1,2,\cdots,2^{n-1})$ are homogeneous,
$H_{n,2^{n-1}}(s_1,s_2,\cdots,s_n)$ is a monic polynomial of degree
$2^{n-1}$, and
\begin{eqnarray}\label{induction-step2}
\sum^{2^{n-1}}_{j=0}H_{n,j}(t_1^2,t_2^2,\cdots,t_n^2)(\delta-t_{n+1})^{2(2^{n-1}-j)}=0.
\end{eqnarray}
Applying Newton binomial theory, we then get that
\begin{eqnarray*}
& &H_{n,n}(t_1^2,t_2^2,\cdots,t_n^2)+\delta^{2^n}+t_{n+1}^{2^n}\\
&
&~~~+\sum^{2^{n-1}-1}_{j=1}H_{n,j}(t_1^2,t_2^2,\cdots,t_n^2)\sum_{l=0}^{2^{n-1}-j}
\begin{pmatrix}2^{n}-2j\\2l\end{pmatrix}\delta^{2^{n}-2j-2l}t^{2l}_{n+1}\\
& &~~~+\sum^{2^{n-1}-1}_{l=1}\begin{pmatrix}2^n\\ 2l\end{pmatrix}\delta^{2^{n}-2l}t^{2l}_{n+1}\\
&=&\delta
t_{n+1}\left(\sum^{2^{n-1}-1}_{j=1}H_{n,j}(t_1^2,t_2^2,\cdots,t_n^2)\sum^{2^{n-1}-j}_{l=1}\begin{pmatrix}2^{n}-2j\\
2l-1\end{pmatrix}\delta^{2^{n}-2j-2l}t^{2l-2}_{n+1}\right)\\
& &+\delta
t_{n+1}\sum^{2^{n-1}}_{l=1}\begin{pmatrix}2^n\\
2l-1\end{pmatrix}\delta^{2^{n}-2l}t^{2l-2}_{n+1}.
\end{eqnarray*}
Taking squares of both sides of the above $=$, we arrive at the
conclusion of Lemma \ref {quadratic-lemma} when $n$ is replaced by
$n+1$. This completes the proof.
\end{proof}

The following conclusion results from  the proof of Theorem 1 of
\cite{pino-bokanowski-etal07} (c.f. also \cite{lieb83}):

\begin{lemma}\label{pino-etal-lemma}
Given $v,v'\in \mathbb{V}$.  Let $\rho_v=\rho^{\Psi_v}$ and
$\rho_{v'}=\rho^{\Psi_{v'}}$ with $\Psi_v\in \mathscr {G}_{v}$ and
$\Psi_{v'}\in \mathscr {G}_{v'}$. If $\rho_v=\rho_{v'}$, then
\begin{eqnarray}\label{distribution-identity}
\left(\sum^N_{i=1}(v'-v)(x_i)-(E(v')-E(v))\right)\Psi_v=0.
\end{eqnarray}
\end{lemma}

\begin{proof}
For completion, we present a proof here, which essentially comes
from the proof of Theorem 1 of \cite{pino-bokanowski-etal07}. We see
that
\begin{eqnarray*}
& &E(v)=(\Psi_v,{\cal {H}}_v\Psi_v)\le (\Psi_{v'}, {\cal
{H}}_v\Psi_{v'})\\
&=& E(v')-\int_{\mathbb{R}^3}\rho_{v'}(v'-v).
\end{eqnarray*}
Similarly,
\begin{eqnarray*}
E(v')\le E(v)-\int_{\mathbb{R}^3}\rho_{v}(v-v'). \end{eqnarray*}
Thus we obtain that if $\rho_v=\rho_{v'}$, then$$
E(v')=E(v)-\int_{\mathbb{R}^3}\rho_{v}(v-v'),$$or
$$
\int_{\mathbb{R}^3}\rho_{v}(v-v')=E(v)-E(v'),
$$
which leads to $E(v)=(\Psi_{v'}, {\cal {H}}_v\Psi_{v'})$. Therefore
$\Psi_{v'}\in \mathscr {G}_{v}$ and
\begin{eqnarray*}
{\cal {H}}_v \Psi_{v'}=E(v)\Psi_{v'}. \end{eqnarray*} By a similar
argument, we have
\begin{eqnarray*} {\cal {H}}_{v'} \Psi_{v}=E(v')\Psi_{v}.
\end{eqnarray*}
Since
\begin{eqnarray*}
{\cal {H}}_v \Psi_{v}=E(v)\Psi_{v},
\end{eqnarray*}
we arrive at (\ref{distribution-identity}). This completes the
proof.
\end{proof}

Due to the Fundamental Theorem of Algebra (c.f., e.g.,
\cite{smithies00}),  every non-zero single-variable polynomial with
real or complex coefficients has exactly as many real or complex
zeroes as its degree, if each zero is counted up to its
multiplicity. Hence we have a multivariate version of the
Fundamental Theorem of Algebra as follows:

\begin{lemma}\label{multi-FTA} The Lebesgue's measure of the set of zeroes of any non-zero multivariate
polynomial with real coefficients is zero.
\end{lemma}

\section{Proof}\label{proof}
\setcounter{equation}{0} In this section, we prove Theorem
\ref{hohenberg-kohn}, the precise and new statement of
Hohenberg-Kohn theorem.
\begin{proof} Let $v,v'\in \mathbb{V}_{C}$.
 Choose
$\Psi_v\in \mathscr {G}_{v}$ and $\Psi_{v'}\in \mathscr {G}_{v'}$
such that $\rho_v=\rho^{\Psi_v}$ and $\rho_{v'}=\rho^{\Psi_{v'}}$.
It is sufficient to prove that $v=v'$ if $\rho_v=\rho_{v'}$.

Note that there exist $m\ge 1, r_j\in \mathbb{R}^3$ and $\alpha_j\in
\mathbb{R}(j=1,2,\cdots,m)$ such that
$$(v'-v)(x)=\sum^m_{j=1}\frac{\alpha_j}{|x-r_j|}.$$
Suppose $v'\not=v$, we have $\alpha_{j_0}\not= 0$ for some
$j_0\in\{1,2,\cdots,m\}$. As a result, if equation
\begin{eqnarray}\label{v-E-identity}
\sum^N_{i=1}(v'-v)(x_i)=E(v')-E(v)
\end{eqnarray}
holds, then there exist non-zero polynomials
$\{H_{n,j}(s_1,s_2,\cdots,s_n): j=0,1,2,\cdots,2^{n-1}\}$ with real
coefficients satisfying the conclusion of Lemma
\ref{quadratic-lemma} with $n=mN$ and
\begin{eqnarray*}
\delta&=&\frac{\alpha_{j_0}}{|x_{1}-r_{j_0}|},\\
t_1&=& E(v')-E(v),\\
\{t_l:l=2,3,\cdots,n\}&=&\left\{\frac{-\alpha_j}{|x_i-r_j|}:
i=1,2,\cdots,N;j=1,2,\cdots,m\right\}\setminus
\{-\delta\}.\end{eqnarray*} Therefore there exists a non-zero
multivariate polynomial $P(s_1,s_2,\cdots,s_n)$ with real
coefficients
 such that
\begin{eqnarray}\label{non-zero-poly}
& &P(|x_1-r_1|^2,\cdots,|x_i-r_j|^2,\cdots,|x_N-r_m|^2)=0,\nonumber\\
& &~~x_i\in \mathbb{R}^3\setminus \{r_j: j=1,2,\cdots,m\},
i=1,2,\cdots,N
\end{eqnarray}
and the set of the solutions $(x_1,x_2,\cdots,x_N)$ of
(\ref{v-E-identity}) are a subset of zeroes of (\ref{non-zero-poly})
and hence of zero measure in $\mathbb{R}^{3N}$ by Lemma
\ref{multi-FTA}. Since
$$\sum^N_{i=1}(v'-v)(x_i)-(E(v')-E(v))$$  is continuous over domain
$$\{(x_1,x_2,\cdots,x_N): x_i\in
\mathbb{R}^3\setminus \{r_j: j=1,2,\cdots,m\}, i=1,2,\cdots,N\}$$ as
a function of $(x_1,x_2,\cdots,x_N)\in\mathbb{R}^{3N}$, we must have
$\Psi_v=0$ almost every where in $\mathbb{R}^{3N}$ from
(\ref{distribution-identity}), which is a contradiction to
$$\sum_{\sigma_1,\sigma_2,\cdots,\sigma_N}\int_{\mathbb{R}^{3N}}|\Psi_v|^2=1.$$
This completes the proof.
\end{proof}

\vskip 0.2cm

{\sc Acknowledgements.} The author would like to thank Prof. X.
Gong and other members in the author's group for their stimulating
discussions and fruitful cooperations that have motivated this work.

\end{document}